\documentclass[a4paper,oneside,reqno]{article}
\usepackage[english]{babel}

\newcommand{\bib}{bib}
\usepackage{cvpr-geom}
\usepackage{cvpr-abbriv}
\usepackage[alwaysadjust,flushright]{paralist}
\usepackage[pagebackref=false,breaklinks=true,a4paper,colorlinks=false,bookmarks=false,baseurl=cmp.felk.cvut.cz/~shekhovt, hyperfigures=true,linkbordercolor={1 1 1},pdfcenterwindow=true,pdffitwindow=true,pdfstartview=FitH,pdfstartpage=1]{hyperref}
\usepackage{alg}
\usepackage{amsthm}
\usepackage{epsfig}
\usepackage{graphicx}
\usepackage{amsmath}
\usepackage{amssymb}

\usepackage{amsfonts}
\usepackage{verbatim}
\usepackage{color}
\usepackage{epsfig}

\DeclareMathOperator{\conv}{conv}

\newcommand{\<}{\langle}
\renewcommand{\>}{\rangle}

\renewcommand{\mid}{\:|\,}
\newcommand{\bb}{\boldsymbol}

\newcommand{\Real}{\mathbb{R}}
\newcommand{\Natural}{\mathbb{N}}

\newcommand{\tab}{{\hphantom{bla}}}
\newcommand{\A}{\mathcal{A}}

\newcommand{\V}{\mathcal{V}}

\newcommand{\E}{\mathcal{E}}

\renewcommand{\L}{\mathcal{L}}
\newcommand{\LL}{\bb{\mathcal{L}}}

\newcommand{\st}{\mbox{s.t. }}

\newcounter{myRomanCounter}

\newcommand{\f}{{\bf f}}

\newcommand{\g}{{\bf g}}
\newcommand{\h}{{\bf h}}

\newtheorem{theorem}{Theorem}

\newtheorem{corollary}{Corollary}
\newtheorem{statement}{Statement}

\theoremstyle{definition}
\newtheorem{definition}{Definition}

\newcommand{\VH}{V\'aclav Hlav\'a\v c}
\newcommand{\AS}{Alexander Shekhovtsov}

\newcommand*{\mypar}[1]{\subsection{#1}}

\newcommand{\K}{{\bb K}}
\def\leftbb{\mathopen{\rlap{$[$}\hskip1.3pt[}}
\def\rightbb{\mathclose{\rlap{$]$}\hskip1.3pt]}}
\def\barvee{\veebar}
\newcommand{\IF}{\mbox{ \rm if }}

\newcommand{\acknowlegment}{\footnote{The work was supported bu EU projects FP7-ICT-247870 NIFTi and FP7-ICT-247525 HUMAVIPS and the Czech project 1M0567 CAK.}\ }																			
\begin{document}
\twocolumn[\parbox{\linewidth}{
\begin{center}
{\sc \huge 
On Partial Opimality by Auxiliary\\[.5cm] Submodular Problems}\\[1cm]
\setlength{\tabcolsep}{1cm}
\begin{tabular}{cc}
\begin{tabular}[t]{l}
\bf \AS \\
\ Msc., engineer \\
\ shekhole@fel.cvut.cz
\end{tabular}
&
\begin{tabular}[t]{l}
\bf \VH \\
\ Prof., head of the \\
\ Center for Machine Perception \\
\ hlavac@fel.cvut.cz
\end{tabular}\\
\end{tabular}\\[1cm]
Center for Machine Perception, Department of Cybernetics\\
Faculty of Electrical Engineering, Czech Technical University in Prague\\
Technicka 2, 166 27 Prague 6, Czech Republic\\
\end{center}
\begin{abstract}

%
\par
In this work, we prove several relations between three different energy minimization techniques. A recently proposed methods for determining a provably optimal partial assignment of variables by Ivan Kovtun (IK), the linear programming relaxation approach (LP) and the popular expansion move algorithm by Yuri Boykov. 
We propose a novel sufficient condition of optimal partial assignment, which is based on LP relaxation and called LP-autarky. We show that methods of Kovtun, which build auxiliary submodular problems, fulfill this sufficient condition. 
The following link is thus established: 
LP relaxation cannot be tightened by IK. For non-submodular problems this is a non-trivial result. In the case of two labels, LP relaxation provides optimal partial assignment, known as persistency, which, as we show, dominates IK.
Relating IK with expansion move, we show that the set of fixed points of expansion move with any ``truncation'' rule for the initial problem and the problem restricted by one-vs-all method of IK would coincide -- i.e. expansion move cannot be improved by this method. In the case of two labels, expansion move with a particular truncation rule coincide with one-vs-all method.

\end{abstract}
\vskip0.5cm
\par
{\bf Keywords:} energy minimization, partial optimality, persistency, max-sum, WCSP, MRF, autarky, LP-relaxation, expansion move.
\vskip0.5cm
}]


\section{Introduction}
\mypar{Energy Minimization}
In this work\acknowlegment we consider minimization problem of the following form:
\begin{equation}\label{energy minimization}
\min_{x\in\LL}\left[ f_0 + \sum_{s\in \V} f_s(x_s) +\sum_{st\in \E}f_{st}(x_{st}) \right] = \min_{x\in\LL} \f(x).
\end{equation}
Here, $\V$ is a finite set and $\E\subset \V\times\V$. A concatenated vector of all variables $x = (x_s| s\in \V)$ is called a {\em labeling}. Variable $x_s$ takes its values in a discrete {\em domain} $\L_s$, called {\em labels}. Labeling $x$ takes values in $\LL$, the Cartesian product of all domains $\L_s$. In this paper all $\L_s$ will have the same number of labels, but may have different associated orderings, \etc. Notation $st$ denotes the ordered pair $(s,t)$ and $x_{st}$ denotes the pair of corresponding variables, $(x_s,x_t)$.
The objective is composed of term $f_0\in \Real$ and functions $f_s \colon \L_s\to \Real$ and $f_{st}\colon \L_s\times\L_t\to\Real$.
\par
The problem~\eqref{energy minimization} is considered in several fields. It is also known as the labeling problem, the Weighted Constraint Satisfaction (WCSP) and for the case of two labels ($|L_s|=2, \forall s$) as the pseudo-Boolean\footnote{Variables $x_s\in\{0,1\}$ are regarded as Boolean in this case and ``pseudo'' emphasize that a real-valued rather than Boolean function of these variables is considered.} optimization~\cite{BorosHammer02}. Our terminology comes from considering probabilistic models in the form of Gibbs distribution.
There is certain difference between problems with two labels and more than two labels, the later will be referred to as {\em multi-label} problems.
%
%
%
\mypar{Partial Optimality}
Energy minimization~\eqref{energy minimization} is an NP-hard problem in general. Techniques which allow us to find a ``part of the optimal'' labeling are of our central interest here. The idea is that it may be possible to fix a part of variables to take certain labels such that  any optimal labeling will provably have the same partial assignment.
\par
More precisely, we consider a subset of variables $\A\subset \V$ and the assignment of labels over this subset $y = (y_s\mid s \in \A)$. The pair $(\A,y)$ is called a {\em strong optimal partial assignment} (strong persistency~\cite{Boros:TR06-probe}), if for any minimizer $x^*$ it holds $x^*_{\A}=y$, where notation $x^*_{\A}$ is the restriction of $x^*$ to $\A$, \ie $(x^*_s\mid s\in \A)$. Likewise, if there exist at least one minimizer $x^*$, for which $x^*_{\A}=y$ holds we say that $(\A,y)$ is a {\em weak optimal partial assignment}.
\par
Two or more strong optimal partial assignments can be combined together, because each of them preserves all optimal solutions. This is not true for weak assignments, even if they assign different variables, -- they may not share any globally optimal solutions in common. However, if we want to find a minimizer of~\eqref{energy minimization} (or at least ``localize'' it as much as possible), a weak optimal partial assignment could be more helpful -- the best one assigns all variables.
%
%
%
%
\mypar{Domain Constraints} The idea of optimal partial assignment naturally extends to constraining a variable to a subset of labels $K_s \subset \L_s$. Let $\A\subset \V$, let $K_s\subset \L_s$, $\forall s\in \A$. Let $\K$ be the Cartesian product of $K_s$, $s\in \A$. We say that a pair $(\A,\K)$ is a {strong (resp. weak) optimal constraint} if $x^*_\A \in \K$ for all (resp. at least one) minimizer $x^*$. This type of constraints is called {\em domain constraints}. Obviously, it includes partial assignment as a special case.
\mypar{Autarkies}
Some domain constraints follow from more specific properties called ``autarkies''. This term occurs in~\cite{Boros:TR06-probe} for two-label problems and we consider its extension~\cite{Shekhovtsov-07-binary-TR} to multi-label problems.
\par
Let $\L_s=\{0,1\dots L\}$ $\forall s\in\V$, $L\in \Natural$. Let $x,y\in \LL$. Define component-wise minimum and maximum of two labellings:
\begin{subequations}
\begin{align}
&(x \wedge y)_s = \min(x_s,y_s),\\
&(x \vee y)_s = \max(x_s,y_s).
\end{align}
\end{subequations}
A pair $(x^{\rm min}\in \LL, x^{\rm max}\in \LL)$ such that $x^{\rm min}\leq x^{\rm max}$ (component-wise) is called a {\em weak autarky} for problem~\eqref{energy minimization}, if
\begin{equation}\label{autarky}
\forall x\in \LL \tab \f((x \vee x^{\rm min}) \wedge x^{\rm max}) \leq \f(x).
\end{equation}
If additionally for any $x\neq (x \vee x^{\rm min}) \wedge x^{\rm max}$ strict inequality
\begin{equation}
\f((x \vee x^{\rm min}) \wedge x^{\rm max}) < \f(x)
\end{equation}
holds, then the autarky is called {\em strong}.
\par
The autarky provides domain constraints with $K_s = [x^{\rm min}_s,\dots, x^{\rm max}_s]$. For any minimizer $x^*$, we have that $\hat x = (x^* \vee x^{\rm min}) \wedge x^{\rm max}$ is a minimizer as well, and $\hat x_s \in K_s$. A strong autarky guarantees additionally that $x^*$ must itself satisfy $x^*_s \in K_s$. Indeed, if it was not true then $\hat x\neq x^*$ and $\f(\hat x)<\f(x^*)$, which is a contradiction. Therefore a weak (resp. strong) autarky provides a weak (resp. strong) domain constraint.
\par
Determining whether a given pair $(x^{\rm min}, x^{\rm max})$ is a strong autarky is an NP-hard decision problem~\cite{Boros:TR06-probe}.
\par
Autarkies can be combined together. A {\em join} of two autarkies $(x^1,x^2)$, $(y^1,y^2)$ is the pair $(x^1 \vee y^1, x^2 \wedge y^2)$. For strong autarkies, the result is a strong autarky and this operation is commutative, associative and idempotent, so that it defines a semi-lattice. 
\begin{proof}
From definition of autarkies, we have
\begin{equation}\label{join:l1}
\f( (((x \vee x^1)\wedge x^2) \vee y^1) \wedge y^2 ) \leq \f((x \vee x^1)\wedge x^2) \leq \f(x)
\end{equation}
Note, that for $x^1 \leq x^2$ we have $(x \vee x^1) \wedge x^2 = (x \wedge x^2) \vee x^1$.
We can rewrite the labeling in the left hand side (LHS) as follows
\begin{equation}
\begin{split}
&((x \vee x^1)\wedge x^2) \vee y^1) \wedge y^2 = \\
&((x \wedge x^2) \vee  (x^1 \vee y^1) ) \wedge y^2 \doteq\\
&\ \ (x \wedge (x^2 \wedge y^2) ) \vee  (x^1 \vee y^1),
\end{split}
\end{equation}
where doted equality holds if $y^2\geq x^1$. This is satisfied for strong autarkies, because it would be a contradiction that all optimal labellings are below $y^2$ and above $x^1$. 
\end{proof}
Thus there exists an autarky, which provides the maximal amount of domain constraints among strong autarkies. It is the join of all strong autarkies.
\par
It is also possible to join ``non-contradictive'' weak autarkies together, but let us leave it aside for now.
\par
We will consider a special cases of autarkies with ``one-side constraints'', of the form $(x^{\rm min},L)$ or $(0,x^{\rm max})$, where $L$ and $0$ represent the labeling with all components equal to $L$ (resp. 0). For such autarkies inequality~\eqref{autarky} simplifies, because $x\vee 0 = x$ and $x \wedge L=x$.
Methods~\cite{Kovtun03,Kovtun04PHD} compute strong autarkies of this form. 
By taking the join of two strong autarkies $(x^{\rm min},L)$ and $(0,x^{\rm max})$ we can obtain a strong autarky $(x^{\rm min},x^{\rm max})$. However, the reverse is not true: if $(x^{\rm min},x^{\rm max})$ is a strong autarky, it does not imply that $(x^{\rm min},L)$ or $(0,x^{\rm max})$ is an autarky. 
And it is the case that other methods (roof-dual~\cite{BorosHammer02} in the case of two-label problem and its multi-label extension~\cite{Shekhovtsov-07-binary-TR}) can find an autarky of the form $(x^{\rm min},x^{\rm max})$, which is not a join of two one-side autarkies. 

%
%
%
%
\mypar{Submodular Problems}
Function $\f$ is called {\em submodular} if
\begin{equation}
\forall x,y\in\LL\ \ \ \f(x\vee y)+\f(x\wedge y) \leq \f(x)+\f(y).
\end{equation}
In the case $\f$ is defined by~\eqref{energy minimization}, it is submodular iff (see \eg~\cite{Werner-PAMI07}) $\forall st\in \E,\ \forall x_{st},y_{st}\in \L_{st} = \L_s \times \L_t$
\begin{equation}\label{submodularity}
f_{st}(x_{st})+f_{st}(y_{st}) \geq f_{st}(x_{st} \wedge y_{st}) + f_{st}(x_{st} \vee y_{st}).
\end{equation}
Minimizing a pairwise submodular function reduces to {\sc mincut} problem~\cite{Ishikawa03},~\cite{DSchlesinger-K2}. Let $\f$ be submodular and $x^*$ be its minimizer. Then we have the following properties:
\begin{subequations}
\begin{align}
\f(x\vee x^*) \leq \f(x),\\
\f(x\wedge x^*) \leq \f(x).
\end{align}
\end{subequations}
They easily follow from submodularity, noting that $\f(x\vee x^*) \geq \f(x^*)$ and $\f(x\vee x^*) \geq \f(x^*)$.
So, in fact, any pair of optimal solutions $(x^{1*},x^{2*})$ is a weak autarky for this problem. Moreover, if we let
\begin{subequations}
\begin{align}
\label{lowest minimizer}
x^{\rm min} = \bigwedge \arg \min_x \f(x),\\
\label{highest minimizer}
x^{\rm max} = \bigvee \arg \min_x \f(x),
\end{align}
\end{subequations}
where $\arg \min$ is the set of minimizers, we see that both $x^{\rm min}$ and $x^{\rm max}$ are minimizers of $\f$ 
and that $(x^{\rm min}, x^{\rm max})$ is a strong autarky for $\f$. In fact, it is the join of all strong autarkies for $\f$. This strong autarky can be determined from a solution of the corresponding {\sc maxflow} problem.

\section{Approach by Kovtun}
In this section, we review techniques~\cite{Kovtun03,Kovtun04PHD} for building autarkies (and hence domain constrains) by constructing auxiliary problems. We take a somewhat different perspective on these results, however, our statements and proofs here are in a sense equivalent to ones given in~\cite{Kovtun03,Kovtun04PHD}. 
%
%
\begin{theorem}\label{T1}
Let $\f = \g + \h$, let $(x^{\rm min}, x^{\rm max})$ be a strong autarky for $\g$ and a weak autarky for $\h$. Then $(x^{\rm min}, x^{\rm max})$ is a strong autarky for $\f$.
\end{theorem}
\begin{proof}
We have
\begin{equation}
\begin{split}
&\f((x \vee x^{\rm min}) \wedge x^{\rm max})  = \\
&\g((x \vee x^{\rm min}) \wedge x^{\rm max})+\h((x \vee x^{\rm min}) \wedge x^{\rm max}) \leq \\
&\g(x) + \h(x),
\end{split}
\end{equation}
and the inequality is  strict if $(x^{\rm min}, x^{\rm max})$ is strong for either $\h$ or $\g$.
\end{proof}
The idea of auxiliary problems is to construct a submodular $\g$, for which, as we know, a strong autarky $(x^{\rm min},L)$ can be found by choosing $x^{\rm min}$ as the lowest minimizer of $\g$, given by~\eqref{lowest minimizer}. The trick is to find such $\g$ that $(x^{\rm min},L)$ is at the same time an autarky for $\h=\f-\g$. The following sufficient conditions were proposed~\cite{Kovtun04PHD}:

\begin{statement}
\label{stat:sufficient_st}
 Let $\h$ satisfy
\begin{subequations}\label{sufficient_st}
\begin{align}
\label{sufficient_st:a}
\begin{split}
& \forall s\in \V,\  x_s\in \L_s,\ \hat x_s \in K_s \\
& 
\tab h_s(x_s \vee \hat x_s) \leq h_s(\hat x_s) 
\end{split}
\\
\label{sufficient_st:b}
\begin{split}
& \forall st\in\E,\ x_{st}\in \L_{st},\  \hat x_{st} \in K_{st}\\
& h_{st}(x_{st} \vee \hat x_{st}) \leq h_{st}(x_{st}).\\
\end{split}
\end{align}
Then for any $x^{\rm min}$ such that $x^{\rm min}_{s}\in K_s$, the pair $(x^{\rm min}, L)$ is a weak autarky for $\h$. If additionally
\begin{equation}
\label{sufficient_st:c}
\begin{split}
& \forall s\in\V,\ \hat x_s\in K_s,\ x_s<\hat x_s\\
& h_s(x_s \vee \hat x_s) < h_s(\hat x_s),
\end{split}
\end{equation}
\end{subequations}
then $(x^{\rm min},L)$ is a strong autarky.
\end{statement}
\begin{proof}
For any $x\in \LL$, summing corresponding inequalities from~\eqref{sufficient_st:a} and~\eqref{sufficient_st:b}, we obtain
\begin{equation}
\begin{split}
\sum_{s}h_s(x_s\vee x^{\rm min}_s) + \sum_{st}h_{st}(x_{st}\vee x^{\rm min}_{st}) \leq \\
\sum_{s}h_s(x_s) + \sum_{st}h_{st}(x_{st}).
\end{split}
\end{equation}
If $x \vee x^{\rm min} \neq x$, then $\exists s\in\V\ x_s<x^{\rm min}_s$ and~\eqref{sufficient_st:c} implies strict inequality.
\end{proof}

Two practical methods were proposed~\cite{Kovtun04PHD} to construct $\g$ and $(K_s\mid s\in \V)$. We first describe a more general approach.
%
\par\noindent
\begin{myalgorithm}{\\Sequential construction of $\g$, $(K_s\ |\ s\in\V)$, \cite{Kovtun04PHD}}\label{IK1}
%
\item Start with $K_s=\emptyset$, $s\in \V$;
\item Find $\g$ such that $\h = \f-\g$ satisfies~\eqref{sufficient_st} and $\g$ satisfies submodularity constraints~\eqref{submodularity}.
\item Find $x^{\rm min} = \bigwedge \arg \min_x \g(x)$;
\item If $x^{\rm min}_s \in K_s$ for all $s\in \V$ then stop.
\item Set $K_s \leftarrow K_s \cup \{x^{\rm min}_s\}$\ \ $\forall s\in \V$ and go to step 2.
\end{myalgorithm}
In step 2 for each edge $st\in\E$ a system of linear inequalities in $g_{st}$ has to be solved. While~\cite{Kovtun04PHD} provides an explicit solution, for our consideration it will not be necessary. When the algorithm stops, $\g$ is submodular and $(x^{\rm min},L)$ is a strong autarky for $\g$ and a weak autarky for $\f-\g$. By Theorem~\ref{T1}, it is a strong autarky for $\f$. It may stop, however, with $x^{\rm min}_s=0$ for all $s$, so that efficiently no constraints are derived. Being a polynomial algorithm it cannot have a guarantee to simplify the problem~\eqref{energy minimization}.
\par
A simpler non-iterative method proposed in~\cite{Kovtun03} is shown in Algorithm~\ref{IK2}. It attempts to identify nodes $s$ where the label $L$ is better than any other label. The constructed auxiliary problem $\g$ has a property that 
\begin{figure}[!ht]
\begin{myalgorithm}{One vs all method, \cite{Kovtun03,Kovtun04PHD}}\label{one-vs-all}\label{IK2}
\item For each $s$ chose such ordering of $\L_s$ that $0 \in \arg \min\limits_{i\neq L} f_s(i)$.
\item Set $g_s = f_s$, $s\in\V$.
\item Set $K_s = \{0,L\}$.
\item Set $g_{st}(i,j) =
\begin{cases}
a_{st}, \tab i=L,\ j =  L, \\
b_{st}, \tab i= L,\ j\neq L, \\
c_{st}, \tab i\neq L,\ j= L, \\
d_{st}, \tab i\neq L,\ j\neq L,
\end{cases}
$\\
where $a_{st},b_{st},c_{st},d_{st}$ are such that $f_{st}-g_{st}$ satisfy \eqref{sufficient_st:b} and submodularity constraints.
One of the solutions is as follows:
\begin{equation}\label{one-vs-all-abcd}
\begin{split}
a_{st} &= f_{st}(L,L),\\
b_{st} &= \min\limits_{j\neq L}f_{st}(L,j),\\
c_{st} &= \min\limits_{i\neq L}f_{st}(i,L),\\
d_{st} &= \min \Big(b_{st}+c_{st}-a_{st}, \min\limits_{i\neq L, j\neq L}\Big[f_{st}(i,j)\\
&+\min\big\{ b_{st}-f_{st}(L,j), c_{st}-f_{st}(i,L)\big\} \Big]\Big).
\end{split}
\end{equation}
\end{myalgorithm}
\end{figure}
its lowest minimizer $x^{\rm min} = \bigwedge \arg \min\g(x)$ is guaranteed to satisfy $x^{\rm min} \in K_s$ $\forall s\in \V$. (because all costs $(g_{st}(i,j) \mid i<L,\, j<L)$  are equal and $g_s(0)\leq g_s(i)\ \forall s\in\V, \forall i<L$, see proof in~\cite{Kovtun03}). Therefore $(x^{\rm min},L)$ is a weak autarky for $\f-\g$ and Theorem~\eqref{T1} applies.
\par
Both methods allow us to choose various orderings of sets $\L_s$. Strong domain constraints derived from various orderings can be then combined. 

\section{LP-autarkies}
In this section we introduce a special subclass of autarkies, which preserve optimal solutions of the LP-relaxation. Unlike with general autarkies, the membership to this subclass is polynomially verifiable. We show that autarkies constructed by algorithms~\ref{IK1},\,\ref{IK2} belong to this subclass. This has useful implications for LP relaxation.
\mypar{LP Relaxation}
Let $\phi(x)$ be a vector with components $\phi(x)_0=1$, $\phi(x)_{s,i} = \leftbb x_s{=}i \rightbb$ and $\phi(x)_{st,ij} = \leftbb x_{st}{=}ij \rightbb$, where $\leftbb \cdot \rightbb$ is $1$ if the expression inside is true and $0$ otherwise. Let $f$ denote a vector with components $f_0$, $f_{s,i} = f_s(i)$ and $f_{st,ij} = f_{st}(ij)$. With respect to components of energy functions we will be using this index and parenthesis notations completely interchangeably. Let $\<\cdot,\cdot\>$ denote a scalar product. Then we can write energy minimization as
\begin{equation}
\min_{x\in \LL} \<f,\phi(x)\>.
\end{equation}
Its relaxation to a linear program is written as
\begin{equation}\label{lprelax}
\geq \min\limits_{\mu\in \Lambda} \<f,\mu\>,
\end{equation}
where $\Lambda$ is the {\em local polytope}. It approximates $\conv \{\phi(x)\mid x\in \LL\}$ from the outside, see \eg~\cite{Werner-PAMI07} for more detail. It is given by the linear constraints
\begin{equation}
\begin{array}{l}
\mu_0 = 1,\\
\mu_{s,i}\geq 0,\tab \mu_{st,ij}\geq 0, \\
\sum\limits_{ij\in \L_{st}}\mu_{st, ij}=1\ \tab \forall st\in\E,\\
\sum\limits_{j\in \L_t}\mu_{st, ij} = \mu_{s, i} \tab \forall i\in \L_s, \,  st\in \E,\\
\sum\limits_{i\in \L_s}\mu_{st, ij} = \mu_{t,j} \tab \forall j\in \L_t, \,  st\in \E.
\end{array}
\end{equation}
Vector $\mu\in\Lambda$ is called a {\em relaxed} labeling.\\
%

%
%
%
\mypar{LP-autarky}
We now extend the notion of autarky to relaxed labellings.
\begin{definition} A binary operation $\barwedge\colon \Lambda \times \LL \to \Lambda$, is defined as follows. Let $y \in \LL$ and $\mu \in \Lambda$. Then $\nu = \mu \barwedge y \in \Lambda$ is constructed as:
\begin{subequations}
\begin{align}
\nu_{s,i} &= \left\{
\begin{array}{lr}
\mu_{s,i},\tab & i<y_s, \\
\sum\limits_{i'\geq y_s}\mu_{s,i'},\tab & i=y_s, \\
0, \tab & i>y_s;
\end{array}\right.\\
\nu_{st,ij} &=
\left\{
\begin{array}{lr}
\mu_{st,ij}, & i<y_s, j<y_t, \\
\sum\limits_{i'\geq y_s}\mu_{st,i'j}, & i=y_s, j<y_t, \\
\sum\limits_{j'\geq y_t}\mu_{st,ij'}, & i<y_s, j=y_t, \\
\sum\limits_{\substack{i'\geq y_s \\ j'\geq y_t} }\mu_{st,i'j'}, & i=y_s, j=y_t, \\
0, & i>y_s \mbox{ or } j>y_t.\\
\end{array}
\right.
\end{align}
\end{subequations}
\end{definition}
By construction, the relaxed labeling $\nu$ has non-zero weights only for labels ``below'' $y$:
$\nu_{s,i}=0$ for $i>y_s$ and the same for pairs $st,ij$. Let us check that $\nu \in \Lambda$.
\begin{proof}
Normalization constraint:
\begin{equation}
\sum_{i}\nu_{s,i} = \sum_{i<y_s}\mu_{s,i} + \sum_{i'\geq y_s}\mu_{s,i'} = \sum_{i}\mu_{s,i} = 1.
\end{equation}
Marginalization constraint:
\begin{equation}
\begin{split}
\sum_{i}\nu_{st,ij} &= \left\{\begin{array}{lr}
\sum\limits_{i<y_s}\mu_{st,ij} + \sum\limits_{i'\geq y_s}\mu_{st,i'j}, & j<y_t,\\
\sum\limits_{\substack{i<y_s\\ j'\geq y_t}}\mu_{st,ij} + \sum\limits_{\substack{i'\geq y_s\\ j'\geq y_t }}\mu_{st,i'j'},& j=y_t,\\
0,& j>y_t,
\end{array}\right.\\
&= \nu_{t,j}.
\end{split}
\end{equation}
\end{proof}
Operation $\nu = \mu \barvee y$ is defined completely similarly, having singleton components
\begin{equation}
(\mu \barvee y)_{s,i} = \left\{
\begin{array}{lr}
\mu_{s,i},\tab & i>y_s, \\
\sum\limits_{i'\leq y_s}\mu_{s,i'},\tab & i=y_s, \\
0, \tab & i<y_s.
\end{array}\right.\\
\end{equation}
\begin{definition}We say that a pair $(x^{\rm min}, x^{\rm max})$ is a weak {\em LP-autarky} for $\f$, if
\begin{equation}\label{LP-autarky}
\forall \mu \in \Lambda\tab  \<f,(\mu \barwedge x^{\rm min})\barvee x^{\rm max} \> \leq \<f,\mu\>.
\end{equation}
If additionally for all $\mu$ such that $(\mu \barwedge x^{\rm min})\barvee x^{\rm max}\neq \mu$ the strict inequality holds then we say that it is a strong LP-autarky.
\end{definition}
\mypar{Properties of LP-autarkies}
\begin{statement}Any weak (resp. strong) LP-autarky is a weak (resp. strong) autarky.
\end{statement}
\begin{proof}
By substituting $\mu = \phi(x)$.
\end{proof}
\begin{statement}Checking whether $(x^{\rm min}, x^{\rm max})$ is an LP-autarky for $\f$ can be solved in a polynomial time.
\end{statement}
\begin{proof}
By construction, $(\mu \barwedge x^{\rm min})\barvee x^{\rm max}$ is a linear map in $\mu$, let us denote it $A\mu$. Inequality~\eqref{LP-autarky} holds iff
\begin{equation}
\min_{\mu \in \Lambda} \<f, \mu - A \mu\> \geq 0,
\end{equation}
which is a linear program. To verify whether $A$ is a strong LP-autarky we need to solve
\begin{equation}\label{verify strong LP-autarky}
\begin{array}{l}
\min \<f, \mu - A \mu\> > 0 \\
\st \left\{
\begin{array}{l}
\mu \in \Lambda,\\
\sum\limits_{s}\sum\limits_{x^{\rm min} \leq i \leq x^{\rm max}}\mu_{s,i} < |\V|.
\end{array}
\right.
\end{array}
\end{equation}
\end{proof}
%
%
\par
%
\begin{statement}If $\f$ is submodular, then
\begin{equation}\label{relax-submodular}
\begin{split}
\forall \mu \in \Lambda,\ \forall y\in\LL \\
\<\mu,f\>+\<\phi(y),f\>  \geq \<\mu \barwedge y, f\>+ \<\mu \barvee y, f\>.
\end{split}
\end{equation}
\end{statement}
\begin{proof}
Scalar products in~\eqref{relax-submodular} are composed of sums of singleton terms and pairwise terms. We first show that sums of singleton terms are equal, expanding singleton terms in the right hand side (RHS):
\begin{equation}\label{relax-submod-p1}
\small
\begin{split}
\sum_s\sum_i\big[(\mu \barwedge y)_{s,i} &+ (\mu \barvee y)_{s,i} \big]f_s(i) = \\
\sum_s\sum_{i<y_s} \mu_{s,i} f_s(i) &+ \sum_s\sum_{i'\geq y_s} \mu_{s,i'} f_s(y_s) + \\
\sum_s\sum_{i>y_s} \mu_{s,i} f_s(i) &+ \sum_s\sum_{i'\leq y_s} \mu_{s,i'} f_s(y_s) = \\
\sum_s\sum_{i} \mu_{s,i} f_s(i) &+ \sum_s \Big(\sum_{i'} \mu_{s,i'} \Big) f_s(y_s) = \\
\sum_{s}\sum_{i}\mu_{s,i}f_{s}(i) &+\sum_{s}\sum_{i}\leftbb i{=y_s}\rightbb f_s(y_s).
\end{split}
\end{equation}
Now consider submodularity constraints:
\begin{equation}\label{pair-submodularity}
\begin{split}
&\forall st\in \E, \forall ij\in \L_{st}, \forall y_{st}\in \L_{st}\\
&f_{st}(ij)+f_{st}(y_{st}) \geq f_{st}(ij\wedge y_{st})+f_{st}(ij\vee y_{st}).
\end{split}
\end{equation}
Multiplying this inequality by $\mu_{st,ij}$ and summing over $ij$, we obtain on the LHS:
\begin{equation}\label{relax-submod-p2L}
\begin{split}
\sum_{ij}\mu_{st,ij}f_{st}(ij) + f_{st}(y_{st}) = \\
\sum_{ij}\mu_{st,ij}f_{st}(ij) + \sum_{ij}\leftbb ij{=}y_{st}\rightbb f_{st}(y_{st})
\end{split}
\end{equation}
and on the RHS:
\begin{equation}\label{relax-submod-p2R}
\begin{split}
\sum_{ij}\mu_{st,ij}\big[ f_{st}(ij\wedge y_{st})+f_{st}(ij\vee y_{st}) \big] = \\
\sum_{ij}\big[ (\mu \barwedge y)_{st,ij}+ (\mu \barvee y)_{st,ij} \big] f_{st}(ij),
\end{split}
\end{equation}
where the equality is verified as follows:
\begin{equation}
\begin{split}
&\sum_{ij}\mu_{st,ij} f_{st}(ij\wedge y_{st}) = \\
&\sum\limits_{\substack{i<y_s \\ j<y_t}}\mu_{st,ij} f_{st}(ij) +
\sum\limits_{\substack{i \geq y_s \\ j<y_t}}\mu_{st,ij} f_{st}(y_s,j) + \\
&\sum\limits_{\substack{i < y_s \\ j\geq y_t}}\mu_{st,ij} f_{st}(i,y_t) +
\sum\limits_{\substack{i \geq y_s \\ j\geq y_t}}\mu_{st,ij} f_{st}(y_{st}) =\\
&\sum_{ij}(\mu \barwedge y)_{st,ij} f_{st,ij}.
\end{split}
\end{equation}
The term with $\barvee$ is rewritten similarly. By summing inequalities~\eqref{relax-submod-p2L} $\geq$ \eqref{relax-submod-p2R} over $st\in\E$ and adding equalities~\eqref{relax-submod-p1} of the singleton terms, we get the result.
\end{proof}
%
%
\begin{statement}\label{relax-projection}
Let $\f$ be submodular and $x^*\in\arg\min\limits_{x}\f(x)$. Then $\forall \mu\in \Lambda$
\begin{subequations}\label{relax-projection-eq}
\begin{align}
\label{relax-projection-wedge}
\<\mu \barwedge x^*, f\> \leq \<\mu,f\>,\\
\<\mu \barvee x^*, f\> \leq \<\mu,f\>.
\end{align}
\end{subequations}
\end{statement}
\begin{proof} Let us show~\eqref{relax-projection-wedge}. For submodular problems LP-relaxation~\eqref{lprelax} is tight. Thus for any $\nu\in\Lambda$ there holds $\<\nu,f\> \geq \f(x^*) = \<\phi(x^*),f\>$.
In particular, for $\nu = \mu \barwedge y$ we have $\<\mu \barwedge y,f\> \geq \<\phi(x^*),f\>$, which when combined with~\eqref{relax-submodular} implies the statement.
\end{proof}
\begin{statement}\label{relax-strong-opt}
Let $(x^{\rm min},L)$ be a strong LP-autarky for $\f$, then:
\begin{equation}\label{relax-strong-opt-eq}
\begin{split}
& \forall s\in \V,\,\forall i < x^{\rm min}_s,\,
\forall \mu^*\in \arg\min_{\mu\in \Lambda}\<\mu ,f\>\ \tab \mu^*_{s,i} = 0.
\end{split}
\end{equation}
\end{statement}
\begin{proof}
Let $\mu^*\in \arg\min_{\mu\in \Lambda}\<\mu ,f\>$ and $\mu^*_{s,i}>0$. Then $\mu^* \barvee x^{\rm min} \neq \mu^*$ and $\f(\mu^* \barvee x^{\rm min})< \f(\mu^*)$, which contradicts optimality of $\mu^*$.
\end{proof}
\mypar{Implications for Algorithms~\ref{IK1},\,\ref{IK2}}
We have already seen in statement~\ref{relax-projection} that for a submodular problem~$\g$, taking $y$ as a minimizer (resp. the lowest minimizer) of $\g$ gives a weak (resp. strong) LP-autarky $(y,L)$. 
Let us show that statement~\ref{stat:sufficient_st} extends to LP-autarkies too.
This would imply that autarkies derived by algorithms~\ref{IK1},\,\ref{IK2} are in fact LP-autarkies for $\f=\g+\h$. 
%
%
\begin{statement}Let $\h$ satisfy inequalities~\eqref{sufficient_st}. Then for any $y\in \LL$ such that $y_s\in K_s$, the pair $(y,L)$ is a weak LP-autarky for $\h$.
\end{statement}
\begin{proof}
Let $\mu\in \Lambda$. From inequality~\eqref{sufficient_st:a} we have
\begin{equation}\label{eq123}
\sum_{s}\sum_{i}((\mu \barvee y)_{s,i}- \mu_{s,i}) h_{s,i} \leq 0.
\end{equation}
Multiplying~\eqref{sufficient_st:b} by $\mu_{st,ij}$ and summing over $ij\in \L_{st}$ and over $st\in\E$ we obtain
\begin{equation}\label{eq124}
\sum_{st}\sum_{ij}\big[ (\mu \barvee y)_{st,ij}- \mu_{st,ij} \big] h_{st, ij} \leq 0.
\end{equation}
Adding~\eqref{eq123} and~\eqref{eq124}, we get:
\begin{equation}
\<\mu \barvee y - \mu, h\> \leq 0,
\end{equation}
which is equivalent to~\eqref{LP-autarky}.
\end{proof}

We have shown that algorithms~\ref{IK1},\,\ref{IK2} derive domain constraints in the form of strong LP-autarkies. We know too that optimal solutions of LP-relaxation will obey domain constraints derived via strong LP-autarkies. 
Note, while algorithms~\ref{IK1},\,\ref{IK2} depend on the ordering of the labels, solutions of the LP-relaxation does not.
Hence,

\begin{corollary}
Let $(K_s \subset \L_s \mid s\in \V)$ be a strong domain constraint derived by Algorithms~\ref{IK1},\,\ref{IK2} \wrt any ordering of sets $\L_s$. Then the set of optimal solutions of LP relaxation with and without these domain constraints would coincide.
\end{corollary}
We proved that LP relaxation cannot be tightened by algorithms~\ref{IK1},\,\ref{IK2}.
It may only be simplified by eliminating all variables which are guaranteed to be $0$ in every optimal solution. This may be useful in practical methods solving LP relaxation.
\par
%
\par
For problems with two labels, the following relation also holds. Let $\Lambda^* = \arg \min_{\mu\in \Lambda} \<f,\mu\>$. Let 
\begin{equation}
\begin{split}
x^{\rm min}_s = \min\{ i \mid \exists \mu^*\in\Lambda^*\ \mu_{s,i}>0\},\\
x^{\rm max}_s = \max\{ i \mid \exists \mu^*\in\Lambda^*\ \mu_{s,i}>0\},
\end{split}
\end{equation}
then $(x^{\rm min}, x^{\rm max})$ is a strong autarky for $\f$. This is the {\em roof-dual} autarky~\cite{BorosHammer02}. Because for any other autarky derived via~algorithms~\ref{IK1} and~\ref{IK2} statement~\ref{relax-strong-opt} holds, we conclude that roof-dual autarky dominates~algorithms~\ref{IK1} and~\ref{IK2}.

\section{Expansion Move}
Expansion move algorithm~\cite{Boykov01} seeks to improve the current solution $x$ by considering a {\em move}, which for every $s\in\V$ either keeps the current label $x_s$ or changes it to the label $k$.
\begin{myalgorithm}{Expansion-Move~\cite{Boykov01}}\label{alpha-exp}
\item Let $x\in \LL$, let $k \in \L$. The {\em move energy} function $\g(z)$ of binary configuration $z\in \{0,1\}^\V$ is defined by:
\begin{equation}
\begin{array}{l}
g_0 = f_0, \tab g_s(0) = f_s(x_s), \tab g_s(1) = f_s(k),\\
g_{st}(1,1) = f_{st}(k,k), \tab
g_{st}(1,0) = f_{st}(k,x_t), \\
g_{st}(0,1) = f_{st}(x_s,k), \tab
g_{st}(0,0) = f_{st}(x_s,x_t).
\end{array}
\end{equation}
\item Let $z^*\in\arg\min\limits_{z}\g(z)$.\\
\item If $\g(z^*) < \g(0)$, assign $x_s\leftarrow \left\{\begin{array}{ll} x_s, & \IF z_s = 0, \\ k, & \IF z_s = 1.\end{array}\right.$
\end{myalgorithm}
If the above procedure is repeated for all labels $k \in \L$ and no improvement to $x$ is found then $x$ is said to be a fixed point of this method.
\par
In the case $\f$ is a metric energy~\cite{Boykov01}, the move energy~$\g$ is submodular for arbitrary $x$ and step 2 is easy.
\begin{statement}
Let $\f$ be metric~\cite{Boykov01}. Let $(x^{\rm min},L)$ be a strong autarky for $\f$ such that $x^{\rm min}_s \in \{0, L\}$, $\forall s\in\V$.
Then for any fixed point $x$ of the expansion-move algorithm there holds
\begin{equation}
x \geq x^{\rm min}.
\end{equation}
\end{statement}
\begin{proof}
Assume $\exists s\in\V$ such that $\ x_s < x^{\rm min}$. Then $\f(x \vee x^{\rm min}) < \f(x)$ and since $x_s\in\{1,L\}$, it is
\begin{equation}
x_s \vee x^{\rm min}_s = \begin{cases}
x_s, & x^{\rm min}_s=1,\\
L, & x^{\rm min}_s=L,\\
\end{cases}
\end{equation}
which is a valid expansion move from $x$ to label $k=L$, strictly improving the energy.
\end{proof}
In the case when a move energy is not submodular, it can be ``truncated'' to make it submodular while still preserving the property that the move does not increase $\f(x)$~\cite{Rother-05-tapestry}.
Let $\Delta_{st} = g_{st}(1,1)+g_{st}(0,0)- g_{st}(0,1)-g_{st}(1,0)$. Pair $st$ is submodular iff $\Delta_{st}<0$.
\begin{definition}
{\em Truncation} $\g'$ of $\g$ is different from~$\g$ only in non-submodular pairwise components of $\g$, which are set as:
\begin{equation}
\begin{array}{l}
g'_{st,00} = g_{st, 00} - \beta_{st}\Delta_{st},\\
g'_{st, 01} = g_{st, 01} + \alpha_{st}\Delta_{st},\\
g'_{st, 10} = g_{st, 10} + (1-\alpha_{st}-\beta_{st})\Delta_{st},\\
g'_{st, 11} = g_{st, 11},
\end{array}
\end{equation}
where $\alpha_{st}$ and $\beta_{st}$ are free parameters, satisfying $\alpha_{st}\geq 0$, $\beta_{st}\geq 0$, $\alpha_{st}+\beta_{st} \leq 1$. 
\end{definition}
It is easy to verify that $\g'$ 
is submodular, and
\begin{equation}\label{t-alpha-exp-improve}
\g(z) - \g(0) \leq \g'(z) - \g'(0),
\end{equation}
saying that increase in $\g$ is no more than increase in $\g'$ when changing from $0$ to $z$.
\begin{proof}[Proof of~\eqref{t-alpha-exp-improve}.]
By construction of $\g'$, for all $st\in \E$ such that $\Delta_{st}>0$, enumerating all $z_{st}$,
\begin{equation}\label{t-exp:eq2}
\begin{array}{l}
g'_{st,00} - g'_{st,00} = 0,\\
g'_{st,01} - g'_{st,00} = g_{st,01} - g_{st,00} + (\alpha_{st}+\beta_{st})\Delta_{st},\\
g'_{st,10} - g'_{st,00} = g_{st,10} - g_{st,00} + (1-\alpha_{st})\Delta_{st},\\
g'_{st,11} - g'_{st,00} = g_{st,11} - g_{st,00} + \beta_{st}\Delta_{st},
\end{array}
\end{equation}
we see that only positive values are added on RHS. It is also seen that the added positive values do only increase with $\beta_{st}$. This means that the truncation with $\beta_{st}>0$ (let's denote it $g^{\alpha,\beta}$) is {\em never better} than the truncation with $\beta=0$ (let's denote it $g^{\alpha}$): $\forall z $
\begin{equation}\label{t-exp:eq3}
\ \g(z) - \g(0) \leq \g^{\alpha}(z) - \g^{\alpha}(0) \leq \g^{\alpha,\beta}(z) - \g^{\alpha,\beta}(0).
\end{equation}
Similarly, the truncation with $\alpha=0,\beta=1$ ($g^{0,1}$) is not better than the truncation $\g^{\alpha,\beta}$:
\begin{equation}\label{t-exp:eq4}
\g^{\alpha,\beta}(z) - \g^{\alpha,\beta}(0) \leq \g^{0,1}(z) - \g^{0,1}(0).
\end{equation}
This is verified by examining components:
\begin{equation}
\begin{split}
g_{st}^{0,1}(z_{st}) - g^{0,1}_{st}(0) - g^{\alpha,\beta}_{st}(z_{st}) +g^{\alpha,\beta}_{st}(0) = \\
\left\{
\begin{array}{ll}
0, & z_{st}=00,\\
\Delta_{st}(1-(\alpha+\beta)), & z_{st}=01,\\
\Delta_{st}(1-(1-\alpha)), & z_{st}=10,\\
\Delta_{st}(1-\beta), & z_{st}=11,\\
\end{array}\right.\\
\geq 0.
\end{split}
\end{equation}
If $z$ is an improving move for $\g^{0,1}$ then it is also an improving move for any truncation.
\end{proof}
We have the following result about Algorithm~\ref{IK2}:
\begin{statement}Let $(x^{\rm min},L)$ be a strong autarky for $\f$ obtained by Algorithm~\ref{IK2}. 
Let $x$ be a fixed point of the expansion-move algorithm with any truncation rule. 
Then
\begin{equation}
x \geq x^{\rm min}.
\end{equation}
\end{statement}
\begin{proof}
We will prove that the statement holds for truncation $(\alpha=0, \beta=1)$. We need to show that for a move from $x$ to $x\vee x^{\rm min}$ the truncated energy  decreases at least as much as does auxiliary problem built by Alg. 2. This can be verified by inspecting pairwise components for the 4 cases $z_{st} = 00,01,10,11$.
\end{proof}

\section{Conclusion}
We propose a novel representation of methods~\cite{Kovtun03,Kovtun04PHD} as deriving domain constraints via LP-autarkies. This allows for comparison with other methods deriving domain constraints in the same form~\cite{Boros:TR06-probe,Shekhovtsov-07-binary-TR} and establishing relations with common methods of (approximate) optimization. We also believe that ``label domination'' condition proposed by~\cite{Desmet-92-dee} can be interpreted in the same framework, allowing for the theoretical comparison and or for the design of combined methods.
\par
Our results open several directions for improvements. A direct improvement to Alg.~2 can be obtained as follows. Alg.~2 constructs a multi-label auxiliary problem, which is equivalent to a two-label problem (since we know that there is a minimizer with $x^*_s\in \{0,L\}$, $\forall s\in\V$). For two label problems, we also know that the autarky constructed by roof-dual dominates the autarky by truncation, so it will be better to set
\begin{equation}
\begin{array}{l}
d_{st} = \min\limits_{i\neq L, j\neq L}\big[f_{st}(i,j)\\
+\min\big\{ b_{st}-f_{st}(L,j), c_{st}-f_{st}(i,L)\big\} \big]
\end{array}
\end{equation}
and solve for roof-dual using reduction to {\sc maxflow}~\cite{Boros:TR91-maxflow}. This would be a non-submodular auxiliary problem.
\par
We can also attempt to construct auxiliary problem with mixed submodular and supermodular terms as in~\cite{Shekhovtsov-07-binary-TR} or design an algorithm which will propose an autarky in some greedy way and then verify it via solving linear program~\eqref{verify strong LP-autarky}.
%
%
%
\nocite{
Shlezinger76,
Wainwright03nips,
kolmogorov05aistas,
KolmogorovWainwright05Optimality,
Werner-TR-05,
Werner-PAMI07,
Boykov01,
BorosHammer02,
Hammer-84-roof-duality,
Boros:TR06-probe,
Desmet-92-dee,
Boros:TR91-maxflow,
Shekhovtsov-07-binary-TR}

\small
\bibliographystyle{cmpproc}
\bibliography{\bib/max-plus-en,\bib/pseudo-Bool,\bib/kiev-en,\bib/shekhovt}

\begin{thebibliography}{10}

\bibitem{BorosHammer02}
E.~Boros and P.L. Hammer.
\newblock Pseudo-boolean optimization.
\newblock {\em Discrete Applied Mathematics}, (123(1-3)):155--225, 2002.

\bibitem{Boros:TR91-maxflow}
E.~Boros, P.~L. Hammer, and X.~Sun.
\newblock Network flows and minimization of quadratic pseudo-{B}oolean
  functions.
\newblock Technical Report RRR 17-1991, RUTCOR, May 1991.

\bibitem{Boros:TR06-probe}
E.~Boros, P.~L. Hammer, and G.~Tavares.
\newblock Preprocessing of unconstrained quadratic binary optimization.
\newblock Technical Report RRR 10-2006, RUTCOR, Apr. 2006.

\bibitem{Boykov01}
Y.~Boykov, O.~Veksler, and R.~Zabih.
\newblock Fast approximate energy minimization via graph cuts.
\newblock {\em IEEE Transactions on Pattern Analysis and Machine Intelligence},
  23(11):1222--1239, Nov. 2001.

\bibitem{Desmet-92-dee}
J.~{Desmet}, M.~D. {Maeyer}, B.~{Hazes}, and I.~{Lasters}.
\newblock {The dead-end elimination theorem and its use in protein side-chain
  positioning}.
\newblock {\em Nature}, 356:539--542, 1992.

\bibitem{Hammer-84-roof-duality}
P.L. Hammer, P.~Hansen, and B.~Simeone.
\newblock Roof duality, complementation and persistency in quadratic 0-1
  optimization.
\newblock {\em Math. Programming}, pages 121--155, 1984.

\bibitem{Ishikawa03}
Hiroshi Ishikawa.
\newblock Exact optimization for {M}arkov random fields with convex priors.
\newblock {\em IEEE Transactions on Pattern Analysis and Machine Intelligence},
  25(10):1333--1336, 2003.

\bibitem{kolmogorov05aistas}
Vladimir Kolmogorov.
\newblock Convergent tree-reweighted message passing for energy minimization.
\newblock In Robert~G. Cowell and Zoubin Ghahramani, editors, {\em AI and
  Statistics}, pages 182--189. Society for Artificial Intelligence and
  Statistics, 2005.

\bibitem{KolmogorovWainwright05Optimality}
Vladimir Kolmogorov and Martin Wainwright.
\newblock On the optimality of tree-reweighted max-product message passing.
\newblock In {\em To appear in 21st Conference on Uncertainty in Artificial
  Intelligence (UAI)}, July 2005.

\bibitem{Kovtun03}
I.~Kovtun.
\newblock {Partial optimal labeling search for a {NP}-hard subclass of (max, +)
  problems}.
\newblock In {\em DAGM-Symposium}, pages 402--409, 2003.

\bibitem{Kovtun04PHD}
I.~Kovtun.
\newblock {\em Image segmentation based on sufficient conditions of optimality
  in {NP}-complete classes of structural labelling problem}.
\newblock PhD thesis, IRTC ITS National Academy of Sciences, Ukraine, 2004.
\newblock In Ukrainian.

\bibitem{Rother-05-tapestry}
Carsten Rother, Sanjiv Kumar, Vladimir Kolmogorov, and Andrew Blake.
\newblock Digital tapestry.
\newblock In {\em CVPR '05: Proceedings of the 2005 IEEE Computer Society
  Conference on Computer Vision and Pattern Recognition (CVPR'05) - Volume 1},
  pages 589--596, Washington, DC, USA, 2005. IEEE Computer Society.

\bibitem{DSchlesinger-K2}
Dmitrij Schlesinger and Boris Flach.
\newblock Transforming an arbitrary minsum problem into a binary one.
\newblock Research Report {TUD-FI06-01}, Dresden University of Technology,
  April 2006.

\bibitem{Shlezinger76}
M.I. Schlesinger.
\newblock Syntactic analysis of two-dimensional visual signals in the presence
  of noise.
\newblock {\em Cybernetics and Systems Analysis}, 12:612--628, 1976.

\bibitem{Shekhovtsov-07-binary-TR}
A.~Shekhovtsov, V.~Kolmogorov, P.~Kohli, V.~Hlavac, C.~Rother, and P.~Torr.
\newblock {LP}-relaxation of binarized energy minimization.
\newblock .\ Research Report CTU--CMP--2007--27, Czech Technical University,
  2008.

\bibitem{Wainwright03nips}
Martin Wainwright, Tommi Jaakkola, and Alan Willsky.
\newblock Exact {MAP} estimates by (hyper)tree agreement.
\newblock In S.~Thrun S.~Becker and K.~Obermayer, editors, {\em Advances in
  Neural Information Processing Systems 15}, pages 809--816. MIT Press, 2003.

\bibitem{Werner-TR-05}
Tom{\'a}{\v s} Werner.
\newblock A linear programming approach to max-sum problem: {A} review.
\newblock Research Report {CTU--CMP--2005--25}, Center for Machine Perception,
  Czech Technical University, Dec. 2005.

\bibitem{Werner-PAMI07}
Tom{\'a}{\v s} Werner.
\newblock A linear programming approach to max-sum problem: {A} review.
\newblock {\em IEEE Transactions on Pattern Analysis and Machine Intelligence},
  29(7):1165--1179, July 2007.

\end{thebibliography}


\begin{thebibliography}{RKKB05}

\bibitem[BH02]{BorosHammer02}
E.~Boros and P.L. Hammer.
\newblock Pseudo-boolean optimization.
\newblock {\em Discrete Applied Mathematics}, (123(1-3)):155--225, 2002.

\bibitem[BVZ01]{Boykov01}
Y.~Boykov, O.~Veksler, and R.~Zabih.
\newblock Fast approximate energy minimization via graph cuts.
\newblock {\em IEEE Transactions on Pattern Analysis and Machine Intelligence},
  23(11):1222--1239, novembre 2001.

\bibitem[Kol05]{kolmogorov05aistas}
Vladimir Kolmogorov.
\newblock Convergent tree-reweighted message passing for energy minimization.
\newblock In {\em AI and Statistics (AISTATS'05)}, pages 182--189. Society for
  Artificial Intelligence and Statistics, 2005.

\bibitem[Kov03]{Kovtun03}
I.~Kovtun.
\newblock {Partial optimal labeling search for a NP-hard subclass of (max, +)
  problems}.
\newblock In {\em DAGM-Symposium}, pages 402--409, 2003.

\bibitem[Kov04]{Kovtun04PHD}
I.~Kovtun.
\newblock {\em Image segmentation based on sufficient conditions of optimality
  in NP-complete classes of structural labelling problem}.
\newblock PhD thesis, IRTC ITS National Academy of Science Ukraine, 2004.
\newblock In Ukrainian.

\bibitem[KW05]{KolmogorovWainwright05Optimality}
Vladimir Kolmogorov and Martin Wainwright.
\newblock On the optimality of tree-reweighted max-product message passing.
\newblock In {\em To appear in 21st Conference on Uncertainty in Artificial
  Intelligence (UAI)}, July 2005.

\bibitem[RKKB05]{Rother-05-tapestry}
Carsten Rother, Sanjiv Kumar, Vladimir Kolmogorov, and Andrew Blake.
\newblock Digital tapestry.
\newblock In {\em CVPR '05: Proceedings of the 2005 IEEE Computer Society
  Conference on Computer Vision and Pattern Recognition (CVPR'05) - Volume 1},
  pages 589--596, Washington, DC, USA, 2005. IEEE Computer Society.

\bibitem[Sch76]{Schlesinger76}
M.I. Schlesinger.
\newblock Syntactic analysis of two-dimensional visual signals in noisy
  conditions.
\newblock {\em Kibernetika, Kiev}, 4:113--130, 1976.
\newblock In Russian.

\bibitem[Wer05]{Werner-TR-05}
Tom{\'a}{\v s} Werner.
\newblock A linear programming approach to max-sum problem: {A} review.
\newblock Research Report {CTU--CMP--2005--25}, Center for Machine Perception,
  K13133 FEE Czech Technical University, Prague, Czech Republic, Dec. 2005.

\bibitem[WJW03]{Wainwright03nips}
Martin Wainwright, Tommi Jaakkola, and Alan Willsky.
\newblock Exact map estimates by (hyper)tree agreement.
\newblock In {\em Advances in Neural Information Processing Systems 15}, pages
  809--816. MIT Press, 2003.

\end{thebibliography}

\end{document}